\documentclass[11pt]{article}
\usepackage[T1]{fontenc}
\usepackage{lmodern}
\usepackage{microtype}
\usepackage{amssymb}
\usepackage[a4paper,margin=1in]{geometry}
\usepackage{authblk}
\usepackage{amsmath,amsthm,mathtools}
\usepackage{enumitem}
\usepackage{algorithm}
\usepackage{algpseudocode}
\usepackage[unicode,hidelinks]{hyperref}
\providecommand{\Description}[1]{}

\newtheorem{thm}{Theorem}[section]
\newtheorem{lem}[thm]{Lemma}
\newtheorem{prop}[thm]{Proposition}

\theoremstyle{definition}
\newtheorem{defn}[thm]{Definition}
\newtheorem{construction}[thm]{Construction}
\theoremstyle{remark}

\newcommand{\E}{\mathbb E}
\newcommand{\Q}{\mathbb Q}
\newcommand{\R}{\mathbb R}
\newcommand{\ind}{\mathbf{1}}
\DeclareMathOperator{\CTR}{CTR}
\DeclareMathOperator{\OPT}{OPT}

\newcommand{\calR}{\mathcal R}

\title{A Factor-6 Truthful-in-Expectation Mechanism for Cascade Position Auctions}
\author[1,2]{Qiuyu Chen}
\author[2]{Baao Xie\thanks{Corresponding author.}}
\affil[1]{Shanghai Jiao Tong University}
\affil[2]{Ningbo Key Laboratory of Spatial Intelligence and Digital Derivative, Ningbo Institute of Digital Twin, Eastern Institute of Technology, Zhejiang, China}
\date{}
\hypersetup{pdftitle={A Factor-6 Truthful-in-Expectation Mechanism for Cascade Position Auctions},pdfauthor={Qiuyu Chen, Baao Xie}}
\begin{document}
\maketitle
\begin{abstract}
Position auctions require assigning advertisers to positions and determining their display order. Under cascade behavior, earlier clicks prevent users from reaching later placements, coupling the click probabilities of selected advertisements. Existing methods achieve constant-factor welfare approximation and monotone logarithmic-factor approximation separately, leaving open whether both guarantees can be attained simultaneously. The difficulty is that welfare-preserving rounding need not maintain the monotonicity of individual click-through rates required for truthful bidding. We resolve this problem by developing a polynomial-time randomized mechanism that achieves a factor-6 welfare approximation and truthfulness in expectation for arbitrary public advertiser--position click probabilities. The mechanism optimizes a reach-weighted linear relaxation and converts its solution into feasible ordered allocations through matching decomposition, probability-bounded grouping, and calibrated attenuation. This construction compensates for cascade dependence and ensures that each advertiser's expected click-through rate equals exactly one sixth of its fractional click coordinate. The resulting identity transfers the monotonicity of the fractional allocation to the implemented mechanism, enabling exact incentive-compatible payments while preserving the welfare guarantee. The allocation lottery and payments are computable in polynomial bit time, with exact sampling in expected polynomial time.
\end{abstract}
\section{Introduction}\label{sec:intro}
Web advertising auctions allocate limited user attention among advertisers with privately known values per click~\cite{Varian07,EOS07}. On search pages and other sponsored-content displays, a platform decides which advertisements users encounter and how prominently they appear. These decisions determine the welfare generated by the allocation, while auction rules determine the incentives for advertisers to report their values truthfully. Designing computationally efficient mechanisms that achieve high welfare and truthful bidding is therefore a central problem in Web advertising markets~\cite{AGM06}. These challenges are particularly relevant to sponsored placements in AI-generated search responses, where platforms may jointly determine which advertisements appear and in what order~\cite{Balseiro26}.

Position auctions capture the dependence of an advertisement's click probability on its assigned position. Classical position-auction models often use separable advertiser and position effects~\cite{Varian07,AGM06}. In the more general setting studied here, the platform chooses a matching between advertisers and the $m$ available positions and also determines the order in which the selected advertisements are rendered. Under cascade behavior, a user examines the ads sequentially and stops after the first click. Each earlier placement then changes the attention available to all later ones. Consequently, both the matching and its rendering order affect the welfare of the display and the click-through rate received by each advertiser.

For the position-auction model with arbitrary public advertiser--position click probabilities, Balseiro et al.~\cite{Balseiro26} obtain a $4/(1-\varepsilon)$ welfare approximation for every fixed $\varepsilon>0$ and, separately, an $O(\log m)$ approximation with a monotone allocation rule. Whether constant-factor welfare approximation and monotonicity can be achieved together remained open. Randomized rounding may preserve welfare without maintaining the monotonicity of each advertiser's actual expected click-through rate. Since earlier selected advertisements affect the reach of later ones, this dependence is the central obstacle to a truthful constant-factor mechanism.

We resolve this question in the public-click-probability model. There are $n$ advertisers and a display capacity $0\le K\le\min(n,m)$. The conditional click probability $p_{ij}\in[0,1]$ for advertiser $i$ in position $j$ is public; advertiser $i$ reports a nonnegative bid $b_i$ for its private value per click. All probabilities and bids are binary-encoded rationals. Write $b=(b_i)$ and let $\OPT(b)$ be the maximum reported welfare over feasible advertiser--position matchings and all their rendering orders. We prove the following result.

\begin{thm}[Main result]\label{thm:main}
For every public rational matrix $p\in[0,1]^{n\times m}$, capacity $0\le K\le\min(n,m)$, and nonnegative rational bid vector $b$, there is a randomized mechanism satisfying the following properties.
\begin{enumerate}[leftmargin=*,itemsep=2pt,topsep=1pt]
\item Every realized allocation is an ordered matching of size at most $K$, and its expected social welfare is at least $\OPT(b)/6$.
\item Each advertiser's expected click-through rate is nondecreasing in its own bid, with the other bids and public data fixed.
\item Exact nonnegative expected payments make truthful reporting a dominant strategy in expectation. A per-click implementation gives nonnegative utility for every realized click outcome under truthful reporting.
\item The mechanism and its payments are computable in deterministic polynomial bit time. Its allocation lottery has a polynomial-size rational description and can be sampled exactly in expected polynomial bit time.
\end{enumerate}
\end{thm}

The mechanism is built around a reach-weighted linear relaxation of the cascade allocation problem. Its solution is converted into a lottery over feasible ordered matchings by a sequence of matching decomposition, probability-bounded grouping, and independent thinning. The thinning probabilities account for the complete selected prefix, producing a telescoping cancellation of the cascade survival factors. As a result, each advertiser's actual expected click-through rate equals exactly one sixth of its fractional click coordinate. The common scale allows the linear welfare objective to govern both approximation and bidding incentives.

This coordinatewise identity supplies both guarantees. The reach-weighted relaxation upper-bounds the optimal welfare over all orderings. Its bid-independent feasible region and uniform scaling of actual click-through rates imply monotonicity in each advertiser's bid, while scaled externality payments yield truthfulness in expectation. The lottery and payments also have polynomial-length rational descriptions, ensuring an exact computational realization.

The result applies to arbitrary public click-probability matrices, including zero and unit entries, and to any capacity within the number of available advertisers and positions. The ordering benchmark remains unrestricted, so the guarantee covers the joint choice of placements and display order. More generally, the construction identifies exact control of expected click coordinates as a useful tool for mechanism design when user attention depends on preceding allocations.

Section~\ref{sec:related} reviews the relevant mechanisms, and Section~\ref{sec:model} specifies the allocation and utility conventions. Section~\ref{sec:overview} then explains the full argument before the proofs in Sections~\ref{sec:LP}--\ref{sec:complexity}. An explicit example and additional encoding details appear in the appendix.

\section{Related Work}\label{sec:related}
\paragraph{Sequential attention and position auctions.}
Cascade models of search behavior were studied by Craswell et al.~\cite{Craswell08}, Kempe and Mahdian~\cite{KM08}, and Aggarwal et al.~\cite{AFMP08}. Farina and Gatti~\cite{FG17} developed algorithms and mechanisms for ad- and position-dependent cascade externalities, while Gatti et al.~\cite{GRSV18} examined richer forms of advertising externalities and their computational consequences. Balseiro et al.~\cite{Balseiro26} introduce contextual position auctions with arbitrary advertiser--position click probabilities and an endogenous display order. Their $4/(1-\varepsilon)$ welfare algorithm and $O(\log m)$ monotone algorithm establish the bounds that our factor-6 truthful mechanism improves upon in the incentive-compatible setting.

\paragraph{Truthful approximation and rounding.}
For single-parameter bidders, monotone allocations support truthful payments~\cite{AT01}. Lavi and Swamy~\cite{LS11} show how fractional optimization and suitable lotteries support truthful approximation, while Dughmi, Roughgarden, and Yan~\cite{DRY16} develop convex-rounding mechanisms. Both frameworks optimize over a range of distributions independent of reported values. In an ordinary matching auction, preserving assignment marginals directly controls each bidder's allocation. Under cascade attention, a bidder's actual CTR also depends on the retained advertisements earlier in the displayed list. Our construction handles this dependence through probability-bounded grouping and an independent retention rule whose prefix factors telescope. The resulting identity preserves every advertiser's click coordinate at the same factor, allowing the linear objective and externality payments to describe the welfare and incentives of the actual randomized allocation.

\paragraph{Other advertising externalities.}
Recent Web research addresses spatial interactions between ads~\cite{AWZ26} and position auctions with heterogeneous ad sizes and capacity constraints~\cite{BBCK26}. These settings expose different dependencies between chosen placements. Here the mechanism simultaneously handles arbitrary advertiser--position probabilities, matching feasibility, a cardinality limit, and sequential attention.

\section{Model and Preliminaries}\label{sec:model}
Let $[h]=\{1,\ldots,h\}$ for every positive integer $h$. There are positive integers $n$ and $m$ of advertisers and available positions, respectively. Advertiser $i\in[n]$ has a private nonnegative value $v_i$ per click and submits a nonnegative bid $b_i$. The platform knows the matrix $p=(p_{ij})$, where $p_{ij}\in[0,1]$ is the probability of a click on advertiser $i$ at position $j$, conditional on the user reaching that ad in the displayed order. The cardinality limit $0\le K\le\min(n,m)$ bounds the number of displayed advertisements.

\begin{defn}[Feasible ordered allocation]\label{def:allocation}
A feasible allocation is an ordered list
$A=(e_1,\dots,e_\ell)$ with $e_r=(i_r,j_r)\in[n]\times[m]$, $\ell\le K$, pairwise distinct advertisers $i_r$, and pairwise distinct positions $j_r$. Its cascade click-through rates are
\begin{equation}\label{eq:ctr}
\CTR_i(A)=\begin{cases}
 p_{i_rj_r}\prod_{s<r}(1-p_{i_sj_s}),&i=i_r\text{ for some }r,\\
0,&i\notin\{i_1,\ldots,i_\ell\}.
\end{cases}
\end{equation}
The reported social welfare is $W(b,A)=\sum_i b_i\CTR_i(A)$, and the full-information benchmark is
\begin{equation}\label{eq:opt}
\OPT(b)=\max_{A\text{ feasible}} W(b,A).
\end{equation}
\end{defn}

The list is chosen before users interact with the content. For a fixed displayed list, each reached ad receives a click with its public probability, and a click ends the scan. Randomness in the mechanism determines the displayed list; click randomness follows the stated cascade model. The chosen positions need not appear in numeric order because the platform also controls their rendering order. Every ordering of the selected advertiser--position pairs is included in the benchmark. Unused positions may be placed after displayed content without affecting any click probability.

For a randomized allocation $A(b)$, define the expected click rate $y_i(b)=\E[\CTR_i(A(b))]$. Let $t_i(b)$ be the expected payment. At true value $v_i$, expected utility is
\begin{equation}\label{eq:utilitydef}
u_i(v_i;b)=v_i y_i(b)-t_i(b).
\end{equation}
Truthfulness in expectation requires
\begin{equation}\label{eq:DSIC}
u_i(v_i;(v_i,b_{-i}))\ge u_i(v_i;(a,b_{-i}))
\end{equation}
for every advertiser $i$, every true value $v_i\ge0$, every alternative report $a\ge0$, and every fixed vector of other reports. The computational results below use binary-rational inputs; the exact expected utility inequalities in~\eqref{eq:DSIC} are proved for all nonnegative rational values and bids.

Write $E=[n]\times[m]$ for the advertiser--position pairs, and abbreviate $p_e=p_{ij}$ and $x_e=x_{ij}$ whenever $e=(i,j)$ and $x$ is an array indexed by these pairs. We order $E$ lexicographically, first by advertiser and then by position. This public convention also fixes the order in which a selected list is displayed by our mechanism. Ties in optimization and in matching decomposition are resolved according to the same labels. For $K=0$, we output the empty allocation with zero payments; the remainder of the construction takes $K\ge1$.

\section{Overview of the Mechanism}\label{sec:overview}
The mechanism has two linked objectives. It must approximate the welfare of the best ordered allocation, and it must preserve monotonicity in each advertiser's \emph{actual expected} click probability. We address the first objective with a linear relaxation and the second with a randomized realization that scales every advertiser's click coordinate by the same factor.

Let $P$ denote the public reach-weighted polytope defined in Section~\ref{sec:LP}. Its points are nonnegative arrays $x=(x_{ij})$ satisfying the advertiser, position, and cardinality constraints, together with a total weighted click budget of one. Define the fractional click coordinate $q_i(x)=\sum_j p_{ij}x_{ij}$. Section~\ref{sec:LP} shows that every ordered allocation has a feasible reach-weighted representation in $P$. Therefore the linear optimum $V(b)=\max_{x\in P}\sum_i b_iq_i(x)$ upper-bounds $\OPT(b)$, even when the benchmark may choose any rendering order.

The main construction turns any rational $x\in P$ into a lottery $\calR(x)$ over feasible ordered matchings. First, a capacity-preserving decomposition writes $x$ as a convex combination of matchings (Section~\ref{sec:decomp}). Each matching is partitioned into groups whose total standalone click probability is at most one. The aggregate budget in $P$ bounds the weighted number of groups, allowing one group to be sampled so that each edge $e$ is included with probability $x_e/3$. The selected group is then independently thinned with a retention probability determined by its original prefix click probabilities. The resulting cascade survival factors telescope, giving each selected edge conditional expected click probability $p_e/2$ (Section~\ref{sec:rounding}). Consequently,
\begin{equation}\label{eq:mainidentity}
\E[\CTR_i(\calR(x))]=\frac{q_i(x)}{6}\qquad(i\in[n]).
\end{equation}

This identity is the bridge to mechanism design. We choose a lexicographically fixed maximizer $x^*(b)$ of the public LP at the reported bids and sample from $\calR(x^*(b))$. The LP optimum and the coordinatewise scaling give both the $1/6$ welfare guarantee and monotonicity of the implemented expected click rates. Scaled externality payments then establish truthfulness in expectation (Section~\ref{sec:mechanism}). Finally, exact rational LP optimization, integral-flow decomposition, and fair-bit sampling provide the polynomial bit-complexity guarantees (Section~\ref{sec:complexity}).

\section{A Reach-Weighted Linear Relaxation}\label{sec:LP}
We begin by translating cascade reach probabilities into a linear benchmark. The relaxation respects matching and cardinality feasibility and imposes an aggregate click-probability budget. Its essential property is an upper bound on the welfare of every feasible ordering.

\begin{defn}[Public polytope]\label{def:polytope}
Write $d=nm$ and define
\begin{equation}\label{eq:P}
\begin{split}
P=\bigl\{x\in\R_{\ge0}^{n\times m}:\quad
 &\sum_j x_{ij}\le1\quad(i\in[n]),\\
 &\sum_i x_{ij}\le1\quad(j\in[m]),\\
 &\sum_{i,j}x_{ij}\le K,\quad\sum_{i,j}p_{ij}x_{ij}\le1\bigr\}.
\end{split}
\end{equation}
For $x\in P$, let $q_i(x)=\sum_j p_{ij}x_{ij}$ and
\begin{equation}\label{eq:lpvalue}
V(b)=\max_{x\in P}\sum_{i=1}^{n}b_i q_i(x).
\end{equation}
We let $x^*(b)$ be the lexicographically smallest maximizer of~\eqref{eq:lpvalue}.
\end{defn}

The polytope is compact, contains the zero matrix, and depends only on the public data $p$ and $K$. The variables $x_{ij}$ can represent the reach probabilities of a particular deterministic allocation even though they are not ultimately the probabilities with which the rounding selects the edges. The following lemma establishes the approximation benchmark for our randomized mechanism.

\begin{lem}[Upper bound on cascade welfare]\label{lem:upper}
For every feasible ordered matching $A$, there exists $x^A\in P$ such that $q_i(x^A)=\CTR_i(A)$ for all $i$. Hence $V(b)\ge\OPT(b)$ for every $b\ge0$.
\end{lem}
\begin{proof}
Let $A=(e_1,\ldots,e_\ell)$ and put
\[
r_r=\prod_{s<r}(1-p_{e_s}),\qquad
x^A_{e_r}=r_r,
\]
with all other coordinates of $x^A$ equal to zero. Because the selected edges form a matching, each row and column contains at most one nonzero entry, and that entry belongs to $[0,1]$. The capacity inequality holds since $\sum_e x_e^A\le\ell\le K$. The final constraint of~\eqref{eq:P} follows by telescoping:
\begin{equation}\label{eq:reachidentity}
\sum_e p_e x_e^A
=\sum_{r=1}^\ell p_{e_r}r_r
=1-\prod_{r=1}^\ell(1-p_{e_r})\le1.
\end{equation}
Thus $x^A\in P$. The unique potential nonzero term $p_{ij}x^A_{ij}$ for advertiser $i$ is precisely its click-through rate in~\eqref{eq:ctr}. Maximizing the resulting welfare identity over all $A$ proves the claimed upper bound.
\end{proof}

The reach weighting is essential because the standalone probabilities of a displayed list need not sum to one, while their reach-weighted sum always does. Lemma~\ref{lem:upper} therefore gives a linear upper bound without restricting the rendering order.

Let $Q$ denote the larger matching polytope obtained by dropping the final constraint $\sum_{ij}p_{ij}x_{ij}\le1$ from~\eqref{eq:P}. The next section decomposes points of $Q$ into feasible matchings; Section~\ref{sec:rounding} then uses the additional click budget in $P$ to control the randomized selection.

\section{Capacity-Preserving Matching Decomposition}\label{sec:decomp}
The next ingredient is a convex decomposition that preserves advertiser, position, and capacity constraints in every integral outcome. We give an explicit polynomial-time algorithm for rational points of $Q$ so that the lottery and its probabilities can subsequently be computed exactly. For a matching $M\subseteq E$, let $\ind_M$ denote its $n\times m$ incidence matrix.

\begin{lem}[Rational matching decomposition]\label{lem:decomp}
For every rational $x\in Q$, one can compute matchings $M_1,\ldots,M_T\subseteq E$ and positive rational weights $\lambda_1,\ldots,\lambda_T$ satisfying
\begin{equation}\label{eq:decomp}
x=\sum_{t=1}^T\lambda_t\ind_{M_t},\qquad
\sum_{t=1}^T\lambda_t=1,\qquad |M_t|\le K,
\end{equation}
where $T\le d+n+m+2$. Both the decomposition and its rational coefficients can be computed in deterministic polynomial bit time.
\end{lem}

\begin{proof}
We first identify the integral vertices of $Q$. Consider a directed circulation network with one node for each advertiser $i$ and position $j$, together with two nodes $s$ and $t$. Include arcs $s\to i$, $i\to j$, $j\to t$, and $t\to s$, with capacities $1$, $1$, $1$, and $K$, respectively. A point $z\in Q$ extends to a feasible circulation by placing $z_{ij}$ on the advertiser--position arcs, the corresponding row and column sums on the outer arcs, and $\sum_e z_e$ on $t\to s$. Conversely, an integral circulation specifies a matching with at most $K$ edges. Integral circulation feasibility, including integer lower bounds on selected arcs, follows from the standard integral-flow property~\cite{Schrijver03}.

We construct the decomposition directly. Maintain a residual pair $(z,w)$, initially $(x,1)$, satisfying
\begin{equation}\label{eq:residual}
w>0,\qquad z/w\in Q.
\end{equation}
Choose an integral matching $M$ on the smallest face of $Q$ containing $z/w$. This is an ordinary integral circulation with the following requirements: edges with $z_e=0$ are forbidden; every row with $\sum_jz_{ij}=w$ and every column with $\sum_i z_{ij}=w$ must be matched; and if $\sum_ez_e=Kw$, the matching must have cardinality $K$. These are integer lower and upper bounds on the corresponding circulation arcs. Because $z/w$ is a fractional circulation satisfying them, the integral-flow theorem guarantees a matching $M$ that does so as well.

We fix $M$ deterministically by choosing the lexicographically smallest feasible incidence vector in the public order. For each edge in order, test whether its indicator can be fixed to zero while retaining integral feasibility; set it to zero if so and otherwise set it to one. At most $d$ circulation-feasibility tests produce the selected matching.

We subtract as much of $M$ as possible while preserving~\eqref{eq:residual}. Specifically, take $\theta$ to be the minimum of the following positive quantities:
\begin{equation}\label{eq:theta}
\begin{gathered}
w,\qquad z_e\quad(e\in M),\\
w-\sum_j z_{ij}\quad(i\text{ unmatched by }M),\\
w-\sum_i z_{ij}\quad(j\text{ unmatched by }M),\\
\frac{Kw-\sum_ez_e}{K-|M|}\quad\text{if }|M|<K.
\end{gathered}
\end{equation}
A candidate in the last three lines is included only when its denominator or its associated unmatched condition applies. Each candidate is positive: $M$ contains only positive residual edges, and every tight row, column, or cardinality constraint is respected by $M$. Record the pair $(\theta,M)$ and update
\begin{equation}\label{eq:update}
z\leftarrow z-\theta\ind_M,\qquad w\leftarrow w-\theta.
\end{equation}
For a matched row, $w-\sum_jz_{ij}$ is unchanged; for an unmatched row, it decreases by $\theta$. The same observation holds for columns. The total cardinality slack decreases by $(K-|M|)\theta$, while the coordinate slacks of edges in $M$ decrease by $\theta$. All bounds in~\eqref{eq:theta} therefore ensure $z\in wQ$ after the update.

Every constraint that was tight remains tight. If the updated $w$ is positive, the minimum in~\eqref{eq:theta} is attained by a constraint that becomes newly tight. The polytope $Q$ has $d+n+m+1$ defining inequalities, so there are at most that many nonterminal iterations and one terminal iteration. At termination $w=0$, and the nonnegative row constraints force $z=0$. Summing~\eqref{eq:update} over the iterations gives~\eqref{eq:decomp}, including $\sum_t\lambda_t=1$.

For completeness, consider the rational encoding. If the initial point has a common denominator $D_x$, every subsequent residual and weight can be represented with a common denominator dividing $D_x\prod_s h_s$, where each $h_s$ is either one or an integer $K-|M_s|\in[1,K]$ arising in~\eqref{eq:theta}. There are only polynomially many steps, so the denominator has at most $\log_2D_x+O((d+n+m)\log(K+1))$ bits. All residuals and weights belong to $[0,1]$ and have polynomial numerator lengths as well. Each step uses polynomial-size circulation tests and rational arithmetic, completing the algorithmic claim.
\end{proof}

Thus each sampled matching can respect the capacity $K$ without assuming that its standalone click probabilities sum to at most one. For $x\in P$, the click budget holds after weighting these matchings by their decomposition coefficients. The next stage converts that average constraint into a valid probability distribution over bounded-probability groups.

\section{Exact Cascade Click-Rate Realization}\label{sec:rounding}
The decomposition alone does not control cascade reach, since a matching $M_t$ may have $\sum_{e\in M_t}p_e>1$. We first split each matching into bounded-probability groups, then use calibrated independent thinning to implement a common click-rate scale. Algorithm~\ref{alg:rounding} defines the full randomized map $\calR(x)$, and the three subsections below prove validity of group selection, the survival correction, and the resulting coordinatewise identity.

\begin{algorithm}[t]
\caption{Exact cascade realization $\calR(x)$}\label{alg:rounding}
\begin{algorithmic}[1]
\Require Rational $x\in P$ and public $(p,K)$.
\Ensure A feasible ordered matching $A$.
\State Compute $x=\sum_{t=1}^T\lambda_t\ind_{M_t}$ as in Lemma~\ref{lem:decomp}.
\For{$t=1,\ldots,T$}
    \State Partition $M_t$ in public order by next-fit into
    \Statex \hspace{\algorithmicindent} indexed groups $C_{t,1},\ldots,C_{t,B_t}$.
\EndFor
\State Sample an indexed group $C=C_{t,k}$ with mass $\lambda_t/3$;
\Statex \hspace{\algorithmicindent} choose $C=\varnothing$ with the remaining mass.
\If{$C=\varnothing$}
    \State \Return $()$.
\EndIf
\State $S\gets 0$, $A\gets()$.
\For{each $e$ in $C$, in public order}
    \State Draw $Z_e\sim\operatorname{Bernoulli}((2-S)^{-1})$
    \Statex \hspace{\algorithmicindent} independently of previous draws.
    \If{$Z_e=1$}
        \State Append $e$ to $A$.
    \EndIf
    \State $S\gets S+p_e$. \Comment{Include deleted edges.}
\EndFor
\State \Return $A$.
\end{algorithmic}
\end{algorithm}

\subsection{Probability-bounded groups}
Process the edges of each matching $M_t$ in increasing public order. Start a new group whenever placing the next edge in the current group would make the sum of its standalone click probabilities exceed one. Write the resulting nonempty groups as
$C_{t,1},\ldots,C_{t,B_t}$, and put
\begin{equation}\label{eq:bin-notation}
T_t:=\sum_{e\in M_t}p_e,\qquad
\sum_{e\in C_{t,k}}p_e\le1.
\end{equation}
An empty matching has zero groups. This is the standard next-fit partition rule with edge weights $p_e\in[0,1]$.

\begin{lem}[Total probability-group mass]\label{lem:bins}
For every matching $M_t$, the next-fit partition satisfies $B_t\le 2T_t+1$. Consequently, for $x\in P$ and any decomposition~\eqref{eq:decomp},
\begin{equation}\label{eq:binmass}
\sum_{t=1}^T\lambda_t B_t\le3.
\end{equation}
\end{lem}
\begin{proof}
Two consecutive nonempty groups have combined probability sum strictly greater than one: the first edge of the later group did not fit in the preceding group. Pair the groups consecutively. When $B_t=2h$, the $h$ disjoint pairs have total mass greater than $h$, and hence $B_t<2T_t$. When $B_t=2h+1$ with $h\ge1$, the first $2h$ groups have mass greater than $h$, giving $B_t<2T_t+1$. The cases $B_t\in\{0,1\}$ obey the stated weak inequality as well.

Multiplying by $\lambda_t$, summing, and using~\eqref{eq:decomp} gives
\[
\sum_t\lambda_t B_t
\le2\sum_t\lambda_t T_t+\sum_t\lambda_t
=2\sum_e p_e x_e+1\le3.
\]
The last inequality is the probability-budget constraint in~\eqref{eq:P}.
\end{proof}

\begin{construction}[Selection of a probability group]\label{con:group}
Select an indexed group $C_{t,k}$ with probability $\lambda_t/3$. With the remaining probability $1-\sum_t\lambda_t B_t/3$, select the empty outcome. Retain the edges of the chosen group in their public order for the subsequent thinning step.
\end{construction}

The distribution is well defined by Lemma~\ref{lem:bins}. An edge $e$ occurs in exactly one group of each matching $M_t$ that contains it. Therefore its probability of reaching this stage is
\begin{equation}\label{eq:edgeprob}
\Pr\{e\text{ belongs to the selected group}\}
=\sum_{t:e\in M_t}\frac{\lambda_t}{3}
=\frac{x_e}{3}.
\end{equation}
All chosen groups have size at most $K$ and use distinct advertisers and positions. Their standalone click probabilities have total at most one, which allows an especially simple survival correction.

\subsection{Independent thinning}
Let $C=(e_1,\ldots,e_\ell)$ be the selected group. Write
\begin{equation}\label{eq:prefix}
S_0=0,\qquad S_r=\sum_{s=1}^r p_{e_s}\le1.
\end{equation}
For each edge $e_r$, independently sample a retention indicator $Z_r$ with
\begin{equation}\label{eq:a}
\Pr(Z_r=1)=a_r:=\frac{1}{2-S_{r-1}}.
\end{equation}
Display the retained subsequence in the original order. The prefix sums $S_r$ are computed from the full chosen group, including edges that the thinning subsequently deletes. Since $S_{r-1}\le1$, each $a_r$ belongs to $[1/2,1]$.

\begin{lem}[Exact conditional click probability]\label{lem:attenuation}
For any selected group $C=(e_1,\ldots,e_\ell)$ with $\sum_{e\in C}p_e\le1$, independent thinning according to~\eqref{eq:a} gives
\begin{equation}\label{eq:half}
\E\bigl[\CTR_{e_r}\mid C\text{ selected}\bigr]
=\frac{p_{e_r}}{2}\qquad(1\le r\le\ell),
\end{equation}
where $\CTR_{e_r}$ denotes the click probability contributed by the original edge $e_r$, taken to be zero when it is deleted.
\end{lem}
\begin{proof}
For a fixed vector $Z$ of retention decisions, the click probability of edge $e_r$ is
\[
Z_r p_{e_r}\prod_{s<r}(1-Z_sp_{e_s}).
\]
Taking expectation over the independent indicators yields
\begin{equation}\label{eq:conditionalCTR}
a_rp_{e_r}\prod_{s<r}(1-a_sp_{e_s}).
\end{equation}
The factors telescope, because~\eqref{eq:prefix} and~\eqref{eq:a} imply
\begin{equation}\label{eq:telescope}
1-a_sp_{e_s}
=1-\frac{S_s-S_{s-1}}{2-S_{s-1}}
=\frac{2-S_s}{2-S_{s-1}}.
\end{equation}
Hence $\prod_{s<r}(1-a_sp_{e_s})=(2-S_{r-1})/2$. Substituting in~\eqref{eq:conditionalCTR} cancels the denominator of $a_r$ and leaves exactly $p_{e_r}/2$.
\end{proof}

The cancellation is the critical property of the mechanism. The selected group's edges have different positions in the sequence and different prefix sums, yet thinning assigns the same multiplicative factor $1/2$ to every edge's conditional CTR. It also works when $p_e=0$ and $p_e=1$, since every denominator in~\eqref{eq:a} lies in $[1,2]$. All thinning decisions occur before the user clicks, so each realization is an ordinary, feasible cascade allocation.

\subsection{The coordinatewise rounding identity}
Algorithm~\ref{alg:rounding} implements the composition of Construction~\ref{con:group} and the independent thinning rule~\eqref{eq:a}. The public lexicographic rules fix the decomposition and the ordering; thus $\calR(x)$ is determined by $x$, $p$, $K$, and the random bits used to sample the lottery.

\begin{prop}[Exact scaled click-through rates]\label{prop:rounding}
For every rational $x\in P$, $\calR(x)$ is a distribution over feasible ordered matchings of size at most $K$, and
\begin{equation}\label{eq:rounding}
\E[\CTR_i(\calR(x))]=\frac16\sum_jp_{ij}x_{ij}
=\frac{q_i(x)}6\qquad(i\in[n]).
\end{equation}
\end{prop}
\begin{proof}
By~\eqref{eq:edgeprob}, an edge $e$ belongs to the selected group with probability $x_e/3$. Conditional on any group containing that edge, Lemma~\ref{lem:attenuation} gives it expected click probability $p_e/2$. Its unconditional contribution is therefore $x_ep_e/6$. Since the edges within every group are a matching, at most one edge for advertiser $i$ can be displayed in any realization. Summing these contributions over all positions $j$ gives~\eqref{eq:rounding}. Group selection, thinning, and deletion preserve the original matching constraints and the cardinality bound.
\end{proof}

The factor six is the product of the group-selection normalization and the cascade correction. The first preserves each edge's selection probability at scale $1/3$; the second equalizes conditional click rates at scale $1/2$. The coordinatewise identity therefore supplies exactly the information needed to derive incentives from linear optimization in the next section.

\section{Welfare and Incentive Guarantees}\label{sec:mechanism}
We now combine the LP upper bound with the exact click-rate realization to prove welfare approximation and truthful bidding. At reported bids $b$, the mechanism optimizes~\eqref{eq:lpvalue} and applies Proposition~\ref{prop:rounding}:
\begin{equation}\label{eq:mechanism}
A(b)\sim\calR(x^*(b)).
\end{equation}
The lottery is determined before the user interacts with the ads. Its construction depends on the bids only through the selected optimizer $x^*(b)$. In particular, all decomposition, grouping, and thinning probabilities are functions of that optimizer and the public click matrix.

\begin{prop}[Constant-factor welfare and monotonicity]\label{prop:welfare}
The allocation rule~\eqref{eq:mechanism} achieves expected welfare
\begin{equation}\label{eq:welfare}
\E[W(b,A(b))]=\frac{V(b)}6\ge\frac{\OPT(b)}6.
\end{equation}
For every advertiser $i$, its expected click probability $y_i(b)$ is nondecreasing in its own report $b_i$, holding $b_{-i},p,K$ fixed.
\end{prop}
\begin{proof}
Using~\eqref{eq:rounding},
\[
\E[W(b,A(b))]
=\sum_i b_i\frac{q_i(x^*(b))}{6}
=\frac{V(b)}6.
\]
Lemma~\ref{lem:upper} gives the last inequality in~\eqref{eq:welfare}.

Fix advertiser $i$, other reports $b_{-i}$, and bids $a<a'$. Write $x=x^*(a,b_{-i})$, $x'=x^*(a',b_{-i})$, and $F(z)=\sum_{h\ne i}b_h q_h(z)$. The optimality of the two LP solutions gives
\begin{align}
a q_i(x)+F(x)&\ge a q_i(x')+F(x'),\label{eq:mon1}\\
a'q_i(x')+F(x')&\ge a'q_i(x)+F(x).\label{eq:mon2}
\end{align}
Adding these inequalities cancels the other bidders' contributions and yields
\begin{equation}\label{eq:monotonicity}
(a'-a)\bigl(q_i(x')-q_i(x)\bigr)\ge0.
\end{equation}
Because $a'-a>0$, we have $q_i(x')\ge q_i(x)$. Proposition~\ref{prop:rounding} gives $y_i(b)=q_i(x^*(b))/6$, so $y_i(a',b_{-i})\ge y_i(a,b_{-i})$. When $a'=a$, our fixed tie-breaking rule selects the same point and lottery.
\end{proof}

The proof illustrates the importance of an identity for each bidder's expected CTR. Optimization over a fixed polytope gives the two inequalities~\eqref{eq:mon1}--\eqref{eq:mon2}; the exact rounding identity connects their coordinates to the click probabilities on which payments are based. We next implement the corresponding incentive-compatible payments without bid integration.

\subsection{Exact externality payments}
For any bid vector $b$, let
\begin{equation}\label{eq:Hi}
H_i(b_{-i})=\max_{x\in P}\sum_{h\ne i}b_hq_h(x)
=V(0,b_{-i}),
\end{equation}
and set the expected payment of advertiser $i$ to
\begin{equation}\label{eq:payment}
t_i(b)=\frac16\left(H_i(b_{-i})
-\sum_{h\ne i}b_hq_h(x^*(b))\right).
\end{equation}
Thus the payment is the loss imposed on the other advertisers in the relaxation, scaled by the same factor used by the allocation lottery. The number $H_i$ is an LP optimum over the existing polytope $P$ with bidder $i$'s coefficient set to zero. Computing all expected payments requires one such additional LP per advertiser.

For per-click charging, compute $y_i(b)=q_i(x^*(b))/6$. If $y_i(b)>0$, charge
\begin{equation}\label{eq:perclick}
\gamma_i(b)=\frac{t_i(b)}{y_i(b)}
\end{equation}
when advertiser $i$ receives a click, and charge zero otherwise. If $y_i(b)=0$, charge zero. The expected total payment equals $t_i(b)$ because the probability of a click on bidder $i$ is $y_i(b)$.

\begin{prop}[Truthfulness in expectation and individual rationality]\label{prop:truthful}
The allocation rule~\eqref{eq:mechanism}, together with payments~\eqref{eq:payment}, is dominant-strategy truthful in expectation for every nonnegative rational value and report. The payments satisfy
\begin{equation}\label{eq:IR}
0\le t_i(b)\le b_i y_i(b)\qquad(i\in[n]).
\end{equation}
With per-click charges~\eqref{eq:perclick}, truthful reporting gives nonnegative realized utility for every click outcome.
\end{prop}
\begin{proof}
The first inequality in~\eqref{eq:IR} follows from the definition of $H_i$: it is the maximum of the other bidders' objective, so
$H_i\ge\sum_{h\ne i}b_hq_h(x^*(b))$.
For the second inequality, nonnegativity of $b_i$ and $q_i$ gives $H_i\le V(b)$. Subtract the same other-bidder objective and divide by six to obtain
\begin{align*}
t_i(b)&\le\frac{V(b)-\sum_{h\ne i}b_hq_h(x^*(b))}{6}\\
&=\frac{b_iq_i(x^*(b))}{6}=b_i y_i(b).
\end{align*}

Now fix a true value $v_i$, reports $b_{-i}$ of the other advertisers, and an arbitrary report $a$ by advertiser $i$. Let $x_a=x^*(a,b_{-i})$. By~\eqref{eq:rounding} and~\eqref{eq:payment}, the expected utility of $i$ is
\begin{equation}\label{eq:utility}
\begin{split}
u_i(v_i;(a,b_{-i}))
=\frac16\biggl(&v_iq_i(x_a)\\[-2pt]
&+\sum_{h\ne i}b_hq_h(x_a)-H_i(b_{-i})\biggr).
\end{split}
\end{equation}
The last term is independent of advertiser $i$'s report. If $i$ reports $a=v_i$, the selected point maximizes the first two terms of~\eqref{eq:utility} over $P$. No alternative report can yield higher expected utility. This establishes~\eqref{eq:DSIC} directly.

The payment bound~\eqref{eq:IR} implies expected individual rationality. It also yields the realized guarantee for per-click charging: if $y_i>0$, then $0\le\gamma_i=t_i/y_i\le b_i$. Under truthful reporting $b_i=v_i$, a click generates utility $v_i-\gamma_i\ge0$, while the absence of a click generates utility zero. When $y_i=0$,~\eqref{eq:IR} forces $t_i=0$, so the prescribed zero charge is consistent with the same conclusion.
\end{proof}

\subsection{Optimization over a fixed lottery range}
The exact click-rate identity also identifies a bid-independent range of distributions. For fixed public inputs, consider the collection
\begin{equation}\label{eq:range}
\mathfrak{D}=\{\calR(x):x\in P\cap\Q^{n\times m}\}.
\end{equation}
Proposition~\ref{prop:rounding} identifies the welfare of each lottery in this range as $\sum_i b_iq_i(x)/6$. Every rational bid vector has a rational LP optimizer, so~\eqref{eq:mechanism} achieves the maximum of expected welfare over $\mathfrak D$. The expected welfare associated with a distribution $\calR(x)$ is therefore a linear function of the same fractional click coordinates used in the LP. Selecting $x^*(b)$ maximizes welfare over this range for the reported bids. The payments~\eqref{eq:payment} are the corresponding externality payments, as in the truthful rounding framework~\cite{LS11,DRY16}.

The distribution over displayed lists may change substantially with the bids, yet the expected click rates remain the same fixed multiple of the optimized fractional coordinates. Thus the welfare and incentive analysis requires only the LP solutions and exact coordinate identity. The next section shows that the associated lottery and payments also admit polynomial-length rational descriptions.

\section{Exact Computation and Bit Complexity}\label{sec:complexity}
The mechanism is constructive only if its LP optimizer, matching decomposition, lottery probabilities, and payments have polynomially bounded descriptions. We establish these bounds in the binary-rational model and give an exact sampler.

Let $L$ denote the total binary encoding length of $n,m,K$, the bids, and all entries of $p$, using a numerator and positive denominator for each rational. Since the input contains $nm=d$ matrix entries, $d\le L$ under any standard explicit encoding. We use two standard algorithmic facts: rational linear programs have polynomial-time exact optimization algorithms~\cite{GLS88}, and feasible integer lower-bounded circulations can be computed by integral network-flow algorithms~\cite{Schrijver03}.

\begin{prop}[Polynomial construction and payment computation]\label{prop:polytime}
For binary-rational input $(n,m,K,p,b)$, one can compute $x^*(b)$, a polynomial-size rational description of $\calR(x^*(b))$, and every expected payment $t_i(b)$ and per-click charge $\gamma_i(b)$ exactly in deterministic polynomial bit time. All rationals produced have polynomial numerator and denominator bit lengths.
\end{prop}
\begin{proof}
Let $d=nm$. The LP~\eqref{eq:lpvalue} has $d$ variables and $d+n+m+2$ constraints, all with rational coefficients of input-bounded size. Standard exact rational LP optimization therefore computes $V(b)$ in polynomial bit time. On the optimum face, minimize the coordinates $x_e$ successively in public order, fixing each preceding coordinate at its attained minimum. This involves at most $d$ additional rational LP solves and produces the lexicographically smallest optimum $x^*(b)$.

We make the encoding bound explicit. Let $D_p$ be the product of the denominators appearing in the probability matrix; $\log_2 D_p\le L$. Multiply the probability-budget row of~\eqref{eq:P} by $D_p$ to obtain an integer row whose entries and right-hand side have at most $O(L)$ bits. Every other constraint has coefficients in $\{0,\pm1\}$ and right-hand side at most $K$. By Cramer's rule and Hadamard's determinant inequality, all coordinates of every vertex of $P$ have numerator and denominator lengths bounded by
\begin{equation}\label{eq:H}
H=10(d+1)\bigl(L+\lceil\log_2(d+2)\rceil+1\bigr).
\end{equation}
Each successive lexicographic optimization exposes a face of the original polytope $P$; its minimizing coordinate is attained at an original vertex and therefore satisfies the same bound. The selected optimum $x^*(b)$ is itself a vertex. A common denominator $D_x$ for its $d$ coordinates can thus be chosen with $\log_2D_x\le dH$.

The decomposition algorithm in Lemma~\ref{lem:decomp} takes at most
$N=d+n+m+2$ steps. At each step, a common denominator for residuals and weights is obtained from $D_x$ by multiplying by at most one new positive integer not exceeding $K$. Hence its logarithm is at most
\begin{equation}\label{eq:Rbound}
R_0=dH+N\lceil\log_2(K+1)\rceil+O(\log(d+2)).
\end{equation}
All numerators are of comparable polynomial length because residuals and weights lie in $[0,1]$. Each matching is found by at most $d$ integral-circulation feasibility calls on a polynomial-size network, so the entire decomposition takes polynomial bit time.

There are at most $N$ matchings, each of size at most $K$, and thus at most $NK$ indexed probability groups. Their categorical probabilities $\lambda_t/3$ have $O(R_0)$-bit rational representations. Within a group, each prefix sum $S_r$ is a sum of at most $K$ rational input probabilities. Its numerator and denominator have $O(L)$ bits when formed over the product of the input denominators, and $a_r=(2-S_{r-1})^{-1}$ has the same polynomial encoding property. A common denominator of all categorical masses can be taken as the product of their denominators and therefore also has polynomial length.

Finally, computing $H_i(b_{-i})$ requires one additional exact LP optimization for each advertiser. The values $q_i(x^*(b))$, $V(b)$, $H_i(b_{-i})$, and $t_i(b)$ are rational sums and products of polynomially many polynomial-size operands. For $y_i(b)>0$, computing $\gamma_i=t_i/y_i$ is one further exact rational division. Thus all outputs and intermediate values admit polynomial binary descriptions. Appendix~\ref{app:bitdetails} records the determinant and denominator bounds in greater detail.
\end{proof}

\subsection{Exact sampling of the allocation lottery}
The mechanism first selects a probability group from a finite categorical distribution. Conditional on that choice, it uses at most $K$ independent Bernoulli variables with rational parameters. We give a fair-bit sampler with finite expected running time and no rounding of the probabilities.

\begin{prop}[Exact fair-bit sampling]\label{prop:sampling}
An allocation from $\calR(x^*(b))$ can be generated using independent fair bits, with at most $K+1$ rational-probability sampling operations. The sampling terminates almost surely in expected polynomial bit time. For any $\delta\in(0,1)$, all operations finish within
\begin{equation}\label{eq:sampling-tail}
\left\lceil\log_2\frac{K+1}{\delta}\right\rceil+1
\end{equation}
rejection trials each with probability at least $1-\delta$.
\end{prop}
\begin{proof}
Express the probabilities of the indexed groups, together with the residual probability of the empty outcome, as integer weights over a common positive denominator $D$. Draw a uniform integer from $\{0,\ldots,D-1\}$ and locate it in the corresponding half-open interval of cumulative weights to choose a group.

To obtain an exact uniform integer modulo any $D\ge2$, draw $k=\lceil\log_2D\rceil$ independent fair bits, interpreted as an integer $U\in\{0,\ldots,2^k-1\}$. If $U<D$, accept it; otherwise repeat. A trial is accepted with probability $D/2^k>1/2$, and conditional on acceptance, the output is exactly uniform. For $D=1$, the output is deterministic. An exact Bernoulli event with probability $u/D$ is obtained by checking whether the accepted uniform integer is less than $u$.

After the categorical choice, independently generate the at most $K$ retention indicators using the same construction with their own rational probabilities. The denominator lengths are polynomially bounded, so each trial has polynomial bit cost. The expected number of trials for each operation is less than two, giving expected polynomial total cost and almost-sure termination. The probability that any one operation needs more than $r$ trials is at most $2^{-r}$. A union bound over at most $K+1$ operations proves~\eqref{eq:sampling-tail}.
\end{proof}

\begin{proof}[Proof of Theorem~\ref{thm:main}]
Lemma~\ref{lem:upper} bounds optimal welfare by the LP. Proposition~\ref{prop:rounding} realizes each LP click coordinate at scale $1/6$, and Propositions~\ref{prop:welfare} and~\ref{prop:truthful} establish welfare approximation, monotonicity, and truthful payments. Propositions~\ref{prop:polytime} and~\ref{prop:sampling} establish deterministic polynomial-time construction and expected polynomial-time exact sampling. Together these results prove all parts of the theorem.
\end{proof}

\section{Conclusion and Discussion}\label{sec:conclusion}
We establish a polynomial-time factor-6 welfare approximation that is truthful in expectation for cascade position auctions with arbitrary public advertiser--position click probabilities. The central construction realizes the fractional click coordinates of a reach-weighted linear relaxation at a common scale in actual ordered allocations. This exact relationship makes the same optimization problem serve both welfare approximation and incentive-compatible payments, despite the effect of preceding advertisements on subsequent clicks.

The mechanism separates assignment feasibility from sequential attention. Matching decomposition enforces the advertiser, position, and capacity constraints; probability-bounded grouping controls the mass of potential clicks; and calibrated thinning equalizes the expected click-rate scale across placements. These components suggest a general approach to randomized allocation when the value received by a participant depends on earlier allocation decisions. For contextual advertising platforms, the resulting mechanism links position-specific click estimates to an explicitly implementable, incentive-compatible allocation rule.

Further work can seek sharper constants through more efficient probability grouping or alternative attenuation rules, and exploit structural patterns in advertiser--position click matrices. Other promising directions include heterogeneous display constraints and more general sequential attention models, where coordinatewise realization may provide a route to new welfare and incentive guarantees.

\appendix
\section{A Worked Example}\label{app:example}
Consider three advertisers and three positions, with $K=3$. Let the diagonal entries of $p$ be $p_{11}=4/5$, $p_{22}=3/5$, and $p_{33}=1/5$, and set the off-diagonal entries to zero. Take $x_{11}=x_{22}=x_{33}=1/2$ with all other coordinates zero. The row, column, and capacity constraints hold, and the click budget equals
\[
\sum_e p_ex_e=\frac12\left(\frac45+\frac35+\frac15\right)=\frac45.
\]
The fractional click coordinates are
\begin{equation}\label{eq:exampleq}
(q_1(x),q_2(x),q_3(x))=(2/5,3/10,1/10).
\end{equation}
A valid decomposition assigns weight $1/2$ to the diagonal matching $M=\{(1,1),(2,2),(3,3)\}$ and weight $1/2$ to the empty matching. Because $p_{11}+p_{22}>1$, next-fit places $(1,1)$ in the first group and $(2,2),(3,3)$ in the second. Both groups have total standalone click probability $4/5$. Each is selected with probability $1/6$; the empty outcome receives the remaining probability $2/3$. Every diagonal edge consequently enters thinning with probability $x_e/3=1/6$.

For the singleton group, the retention probability is $1/2$. In the second group, the two retention probabilities are $a_1=1/2$ and $a_2=1/(2-3/5)=5/7$. Their conditional click rates are
\[
\frac12\cdot\frac35=\frac3{10},\qquad
\frac57\cdot\frac15\left(1-\frac12\cdot\frac35\right)=\frac1{10}.
\]
Multiplying by the probability $1/6$ of selecting each nonempty group gives the final expected click coordinates

\begin{equation}\label{eq:examplectr}
\left(\E[\CTR_1],\E[\CTR_2],\E[\CTR_3]\right)
=\left(\frac1{15},\frac1{20},\frac1{60}\right)
=\frac16\left(\frac25,\frac3{10},\frac1{10}\right).
\end{equation}

Displaying the full matching $M$ without attenuation would instead give CTRs $4/5$, $(3/5)(1/5)$, and $(1/5)(1/5)(2/5)$ in its public order. Their survival factors differ across positions. Grouping controls the total probability mass, while thinning removes precisely this variation, producing $\E[\CTR_i]=q_i(x)/6$ for every advertiser. The example also separates the roles of the two constraints: matching decomposition enforces advertiser, position, and capacity feasibility, while the click budget bounds the weighted number of groups in the overall lottery.

\section{Additional Encoding Details}\label{app:bitdetails}
We record the bit-length estimates underlying the exact LP, decomposition, and sampling algorithms of Section~\ref{sec:complexity}.

\paragraph{LP vertices and lexicographic optimization.}
Let $D_p$ be the product of the probability denominators, so $\log_2D_p\le L$. Clearing the probability-budget row yields integer coefficients of $O(L)$ bits. All other defining rows have coefficients in $\{0,\pm1\}$ and right-hand sides at most $K$. For an active nonsingular $d\times d$ system, the ordinary augmented row norms are at most $\sqrt{d+K^2}$ and the budget-row norm is at most $D_p\sqrt{d+1}$. Hadamard's inequality and Cramer's rule bound vertex numerators and denominators by $2^{O(L+d\log(d+2))}$; the explicit $H$ in~\eqref{eq:H} dominates this estimate. Each successive lexicographic minimization is over a face of $P$, with its coordinate minimum attained at an original vertex. The final unique optimum is also a vertex. Its objective value is a sum of $d$ products of polynomial-length rationals.

\paragraph{Decomposition denominators.}
If $D_x$ is the product of the denominators of the $d$ coordinates of $x^*$, then $\log_2D_x\le dH$. Suppose the current residual and weight share a denominator $D$. All step sizes in~\eqref{eq:theta} have denominators dividing $D$, except the cardinality candidate, whose denominator divides $D(K-|M|)$. After each subtraction a common denominator thus has the form $Dh$ for some integer $1\le h\le K$. Over at most $N$ steps it divides $D_x\prod_{s=1}^{N}h_s$ and has at most $dH+N\lceil\log_2(K+1)\rceil$ bits. The numerators have comparable length since residual coordinates and weights lie in $[0,1]$. Each matching is computed by lower-bounded integral circulation with capacities at most $K$, using at most $d$ feasibility tests.

\paragraph{Sampling probabilities.}
There are at most $N$ matchings and $NK$ groups. Their categorical masses $\lambda_t/3$, together with the remaining empty-outcome mass, admit a common denominator of polynomial bit length. Each prefix sum $S_r$ uses a subset of the original probabilities, so its denominator divides the product of the input denominators and has $O(L)$ bits. The same holds for $a_r=(2-S_{r-1})^{-1}$ because $2-S_{r-1}\in[1,2]$. Lemma~\ref{lem:bins} makes the empty-outcome mass nonnegative, and every Bernoulli parameter lies in $[0,1]$.

\end{document}